\documentclass[reqno]{amsart}
 \usepackage{cleveref}
 \usepackage{amsaddr}
\crefformat{section}{\S#2#1#3} 
\usepackage{amssymb}
\newtheorem{theorem}{Theorem}[section]
\newtheorem{lemma}[theorem]{Lemma}

\usepackage{etoolbox}
\usepackage{xcolor}
\usepackage{amsmath}
\theoremstyle{definition}
\newtheorem{definition}[theorem]{Definition}
\newtheorem{example}[theorem]{Example}

\newtheorem{prop}[theorem]{Proposition}
\newtheorem{cor}[theorem]{Corollary}

\theoremstyle{remark}
\newtheorem{remark}[theorem]{Remark}
\newtheorem{rmk}[theorem]{Remark}

\numberwithin{equation}{section}

\newcommand{\mf}{\mathfrak}

\newcommand{\C}{\mathbb{C}}

\newcommand{\Z}{\mathbb{Z}}

\newcommand{\lspan}{\mathrm{span}}
\newcommand{\ad}{{\rm ad}}

\newcommand{\Hom}{{\rm Hom}}

\newcommand{\Ind}{{\rm Ind}}
\newcommand{\cdeg}{{\rm cdeg}}

\newtoggle{comments}
\newtoggle{details}
\newtoggle{prelimnote}
\newtoggle{detailsnote}

\iftoggle{comments}{%
	\newcommand{\comments}[1]{
		\ \\
		{\color{red}
			\textbf{Comment:} #1
		}
		\ \\
	}
}{%
\newcommand{\comments}[1]{}
}

\iftoggle{details}{%
	\newcommand{\details}[1]{
		\ \\
		{\color{blue}
			\textbf{Details:} #1
		}
		\ \\
	}
}{%
\newcommand{\details}[1]{}
}

\iftoggle{prelimnote}{%
	
}{%

}

\begin{document}

\title{A remark on semi-infinite cohomology}

\author{Xiao He}

\address{D\'epartement de math\'ematiques et de statistique, \\ Universit\'e Laval,  Qu\'ebec (QC),  Canada}
\email{xiao.he.1@ulaval.ca}


\begin{abstract}
	We extend the notion of semi-infinite cohomology of Lie algebras to include cases where the Lie algebra does not admit a semi-infinite structure but satisfies a mild condition.   Our construction clarifies the definition of affine W-algebras in general nilpotent elements case given in \cite{Kac&Wakimoto2003}. We will also give a characterization of admissible pairs with respect to a nilpotent element in a semisimple Lie algebra and define affine W-algebras associated to admissible pairs, while finite W-algebras associated to admissible pairs were already introduced in \cite{Sadaka}. 
\end{abstract}

\maketitle

\section*{Introduction}
The notion of semi-infinite cohomology (or BRST cohomology) is the mathematical counterpart of BRST reduction in physics. It was introduced by B. Feigin \cite{Feigin 1984} in 1984 for Lie algebras. A super version was studied soon after by A. Kirillov \cite{Kirillov}.  See also \cite{FGZ 1986, A.Voronov1993} for further clarification and explicit calculations. Later on, the technique of semi-infinite cohomology was also studied for associative algebras \cite{Arkhi1} and quantum groups \cite{Arkhi2} by S. Arkhipov. In \cite{IgorAnton},  I. Frenkel and A. Zeitlin realized the quantum group $SL_q(2)$ as a semi-infinite cohomology of the Virasoro algebra with coefficients in a tensor product of two Fock representations.  

Unlike ordinary Lie algebra cohomology,  computing semi-infinite cohomology requires a semi-infinite structure on the Lie algebra. Roughly speaking, a semi-infinite structure is a Lie algebra module structure on the space of semi-infinite forms. The requirement of such a structure is to make sure that the BRST differential is nilpotent, i.e., of square zero, which is essential in cohomology theory. \medskip 

What about if the Lie algebra admits no semi-infinite structure? One way to adjust this is to consider some one-dimensional central extension, which is called  cancellation of anomalies in physics. Another way is, as the physicists did, to add more ``ghosts'', hence to modify the BRST complex,  and then to make a deformation of the BRST differential to make it nilpotent \cite{Bershadsky}.

The present paper will explain which ``ghosts'' should be added, how to modify and to characterize the modified BRST differential in a rigorous mathematical way. As a byproduct, we will give a uniform definition of affine W-algebras in general nilpotent elements case, which clarifies the construction in \cite{Kac&Wakimoto2003}. We will also introduce affine W-algebras associated to admissible pairs, while finite W-algebras associated to admissible pairs were already introduced in \cite{Sadaka}. \medskip

Let $L=\bigoplus_{n\in \Z}L_n$ be a $\Z$-graded Lie algebra with $\dim L_n<\infty$. 
Let $\{e_i~|~i\in \Z \}$ be a well-ordered basis of $L$, where we assume that each $e_i\in L_m$ for some $m$, and if $e_i\in L_m$, then $e_{i+1}\in L_m$ or $e_{i+1}\in L_{m+1}$. Let $L^*=\bigoplus_{n\in \Z}\Hom_{\C}(L_{-n}, \C)$ be the restricted dual of $L$ with dual basis $\{e_i^*~|~i\in \Z \}$. 
A \emph{semi-infinite form} on $L$ is a linear combination of infinite wedge products of $L^*$ of the following type,  \begin{align}\label{eq:infinitewedgeproduct}
\omega=e_{i_1}^*\wedge e_{i_2}^*\wedge\cdots, \end{align}
such that for some $N\in \Z_+$,  we have $i_{k+1}=i_k-1$ for all $k>N$. Let $\ad^*$ be the coadjoint action of $L$ on $L^*$.  A \emph{semi-infinite structure} on $L$ is an $L$-module structure on $\Lambda^{\infty/2+\bullet}L^*$, the space of semi-infinite forms,  under the action
\begin{align}\label{semifiniteactions}
x\cdot e_{i_1}^*\wedge e_{i_2}^*\wedge\cdots :=\sum_{k\geq 1}e_{i_1}^*\wedge \cdots \wedge \ad^*x(e_{i_k}^*) \wedge\cdots
\end{align} 
for $x\in L$ (we need a modification of (\ref{semifiniteactions}) when $x\in L_0$).  When $L$ admits a semi-infinite structure, given a smooth  $L$-module $M$, the author of \cite{A.Voronov1993} defined a square zero differential $d$  on  $M\otimes \Lambda^{\infty/2+\bullet}L^*$. The cohomology of $(M\otimes \Lambda^{\infty/2+\bullet}L^*, d)$ is called  \emph{semi-infinite cohomology} of $L$ with coefficients in $M$.\medskip

It was also shown in \cite{A.Voronov1993} that there is a well-defined 2-cocyle $\gamma$ on $L$, such that $L$ admits a semi-infinite structure if and only if $\gamma\equiv 0$.
In the present paper, we consider the case $\gamma\neq 0$ hence $\Lambda^{\infty/2+\bullet}L^*$ admits no $L$-module structure.  

Let $\ker \gamma:=\{x\in L~|~ \gamma(x, L)\equiv 0\}$ be the radical of $\gamma$, and $F$ a graded complement  of $\ker\gamma$ in $L$. 
Consider the 1-dimensional central extension of $F$ determined by $\gamma$, i.e., the Lie algebra $F\oplus \C K$ with $[x, y]:=-\gamma(x, y)K$ for $x, y\in F$ and $[K, F]\equiv 0$.  Let $\mf F$  be the Fock module of $F$ defined by (\ref{eq:fockrepres}). We show that when  $[L, L]\subseteq \ker \gamma$ is satisfied,  the tensor product  $\Lambda^{\infty/2+\bullet}L^*\otimes \mf F$ admits an $L$-module structure even though  $\Lambda^{\infty/2+\bullet}L^*$ does not.  Moreover, given a smooth $L$-module $M$, the operator $\bar{d}$ on  $ M\otimes\Lambda^{\infty/2+\bullet}L^*\otimes \mf F$ defined by (\ref{def:dbar}) is of square zero. That means, this Fock module $\mf F$ are exactly the ``ghosts'' that we should add in the BRST reduction! 
\medskip

This paper is organized as follows. In \cref{sec:1}, we give a brief review of semi-infinite structure and semi-infinite cohomology. We also show that the affinization of a nilpotent Lie algebra admits a semi-infinite structure.  In \cref{sec:2}, we present an adjustment of semi-infinite cohomology when the Lie algebra admits no semi-infinite structure, and give a characterization of the adjusted differential.  As an application,  we give a uniform definition of affine W-algebras associated to good $\Z$-gradings and also introduce affine W-algebras associated to admissible pairs in \cref{sec:3}. 

All vector spaces, algebras and tensor products are considered over the complex numbers $\C$ except explict declaration. 
\medskip

\noindent {\em Acknowledgements:} The author would like to thank the China Scholarship Council (File No.201304910374) and l'Institut des sciences math\'ematiques for their financial support during the preparation of this paper. The author is also grateful for funding received from the NSERC Discovery Grant of his research supervisor (Michael Lau).

\medskip

\section{Semi-infinite structure and semi-infinite cohomology}\label{sec:1}

A Lie (super)algebra $L$ is called {\em quasi-finite $\Z$-graded} if  $$L=\bigoplus_{n\in \Z}L_n \mbox{~~with~~} \dim L_n<\infty, \mbox{~~and~~} [L_n, L_m]\subseteq L_{m+n} ~\mbox{~for all ~}~ m, n\in \Z .$$  We have  $L=L_{\leq 0}\oplus L_+$, where $L_{\leq 0}:=\bigoplus_{n \leq 0} L_{n}\, \mbox{~and~} \, L_{+}:=\bigoplus_{n>0} L_{n}$  are both subalgebras. 
The $\Z$-grading on $L$ induces $\Z$-gradings on $U (L_{\leq 0}) $, $U (L_{+})$ and  $U(L)$, where $U(-)$ is the universal enveloping algebra functor.  By the PBW theorem,   we have $U (L)\cong
U (L_{\leq 0}) \otimes U (L_{+})$ as vector spaces. An homogeneous element of $U(L)$ is of the form  $\sum_{i=1}^{r} u_{i} v_{i}$ with $u_{i}
\in U (L_{\leq 0}),$ $v_{i} \in U (L_{+})$ and
$\deg (u_{i} v_{i})=\deg (u_{j} v_{j})  \mbox{~for all ~}~ i, j$.
When infinite sums $\sum_{i=-\infty}^{\infty} u_{i} v_{i}$, such that only a finite number of $v_i$ have degree less than a given $N\in \Z_{\geq 0}$ are allowed, we get the {\em completion} $U(L)^{com}$ of $U(L)$.
Products are well-defined in $U(L)^{com}$, which makes it into an associative algebra and $U(L)$ can be considered as a subalgebra of $U(L)^{com}$.

\begin{definition} Let $L$ be a quasi-finite $\Z$-graded Lie algebra. An $L$-module $M$ is called {\em smooth} if for any given  $m\in M$, we have $U(L)_n \cdot m=0$ for $n\gg 0$.
\end{definition}

\begin{rmk}
The completion  $U(L)^{com}$ acts on smooth $L$-modules. Let $M_1, M_2$ be smooth modules for $L_1, L_2$, respectively. Then $M_1\otimes M_2$ is a smooth $L_1\oplus L_2$-module. 
\end{rmk}

\begin{definition}
	Let $(L_1, \circ_1), (L_2, \circ_2)$ be two associative or Lie superalgebras, and $\varphi: L_1\rightarrow L_2$ a homomorphism. A {\em superderivation}  of parity $i\in \Z_2$ with respect to $\varphi$ is a parity-preserving linear map  $D: L_1
	\rightarrow L_2$ satisfying Leibniz's rule
	\begin{align}\label{derivationLie}
	D(u\circ_1 v)=D(u)\circ_2 \varphi(v)+(-1)^{i\cdot p(u)} \varphi(u)\circ_2D(v)
	\end{align}
	for all $u, v \in L_1$ with $u$ homogeneous, where $p(u)$ is the parity of $u$. We call $D$ even if $i=0$ and odd if $i=1$.  
\end{definition}

\begin{rmk}Let $S$ be a generating subset of $L_1$. Then a linear map $D$ satisfying (\ref{derivationLie}) for all $u, v\in S$ can be extended uniquely, through Leibniz's rule, to a superderivation, i.e.,  a  superderivation is completely determined by its value on $S$.  
	
\end{rmk}

\subsection{Semi-infinite structure}
Let $L=L_{\leq 0}\oplus L_{+}$ be a quasi-finite $\Z$-graded Lie algebra.  Let $\{e_i~|~i\leq 0\}$ and $\{e_i~|~i>0\}$ be bases of $L_{\leq 0}$ and $L_+$, respectively, such that each $e_i\in L_m$ for some $m\in \Z$. We also require that whenever $e_i\in L_m$, we have $e_{i+1}\in L_m$ or $e_{i+1}\in L_{m+1}$.  
Let $L^*=\bigoplus_{n\in \Z}L^*_n$ be the restricted dual of $L$ with dual basis $\{e_i^*~|~i\in \Z\}$  such that $\langle e_i^*, e_j\rangle=\delta_{i, j}$, where $L^*_n:=\Hom_{\C} (L_{-n}, \C)$.  

\begin{definition}\label{def:semiinfiniteforms}
	The  space $\Lambda^{\infty/2+\bullet}L^*$ of {\em semi-infinite forms}  on $L$ is the vector space  spanned by infinite wedge products of $L^*$ of the following type,  
	$$\omega=e_{i_1}^*\wedge e_{i_2}^*\wedge\cdots $$
	such that there exists an integer $N(\omega)$ and  $i_{k+1}=i_k-1$ for all $k>N(\omega)$.  
\end{definition}
\bigskip 
Let $\iota(L)$ and $\varepsilon(L^*)$ be copies of $L$ and $L^*$.
For $x\in L$ and $y^*\in L^*$, we denote by $\iota(x)$ and $\varepsilon(y^*)$ the corresponding elements in $\iota(L)$ and  $\varepsilon(L^*)$, respectively. Let
\begin{align*}
cl(L):=\iota(L)\oplus \varepsilon(L^*)\oplus \C K\end{align*} with $\iota(L)\oplus \varepsilon(L^*)$ being odd (note that we assume that $L$ is a Lie algebra, hence purely even), $K$ being even, and with Lie superbracket: for $x, y\in L$ and $u^*, v^*\in L^*$,
\begin{align*}
[\iota(x), \iota(y)]=[\varepsilon(u^*), \varepsilon(v^*)]=0, \quad [\iota(x),\varepsilon(u^*)]=\langle u^*, x\rangle K,\quad  [K, cl(L)]=0.
\end{align*}
The Lie superalgebra $cl(L)$ inherits a natural $\Z$-grading from $L$ with $$cl(L)_n=\begin{cases}
\iota(L_n)\oplus \epsilon(L^*_n) & \mbox{~if~} n\neq 0,\\
\iota(L_0)\oplus \epsilon(L^*_0)\oplus \C K &\mbox{~if~} n=0,
\end{cases}$$ 
and it acts on $\Lambda^{\infty/2+\bullet}L^*$ in the following way: $K$ acts as identity, and for $e_{i_0}\in L$, 
\begin{align*}
\varepsilon(e_{i_0}^*)\cdot e_{i_1}^*\wedge e_{i_2}^*\wedge\cdots &=e_{i_0}^*\wedge e_{i_1}^*\wedge e_{i_2}^*\wedge\cdots, \\
\iota(e_{i_0})\cdot e_{i_1}^*\wedge e_{i_2}^*\wedge\cdots &=\sum_{k\geq 1}(-1)^{k-1}\langle e_{i_k}^*,  e_{i_0}\rangle e_{i_1}^*\wedge\cdots \wedge \hat{e}_{i_k}^*\wedge \cdots.
\end{align*}
The \emph{Clifford algebra} $Cl(L\oplus L^*)$  is defined to be the quotient of $U(cl(L))$ by the ideal generated by $K-1$, and it also has a well-defined action on $\Lambda^{\infty/2+\bullet}L^*$.
\medskip 

For a subspace $V$ of $L$, we let $V^\perp:=\{w^*\in L^*~|~ \langle w^*, u\rangle=0 \mbox{~for all~} u\in V\}$. Then $L_+^\perp = \bigoplus_{n\geq 0}L_n^*$. Let  $\omega_0=e_0^*\wedge e_{-1}^*\wedge e_{-2}^*\wedge \cdots$. Then we have  \begin{align}\label{relation:Fockmodule}
\iota(v)\cdot \omega_0=\varepsilon(u^*)\cdot \omega_0=0, \mbox{~for all~} v\in L_+ \mbox{~and~} u^*\in L_+^\perp. 
\end{align}  
\begin{remark}The elements $\iota(v), \varepsilon(u^*)$ with $v\in L_+ \mbox{~and~} u^*\in L_+^\perp$ are called {\em annihilation operators}.  Note that annihilation operators always anticommute. 
\end{remark}
One can show that  $\Lambda^{\infty/2+\bullet}L^*$ is an irreducible module of $Cl(L\oplus L^*)$ generated by the ``vacuum'' vector $\omega_0$,  with relations defined by  (\ref{relation:Fockmodule}). Every element of $\Lambda^{\infty/2+\bullet}L^*$ can be written as a linear combination of monomials of the form $$\iota(e_{i_1})\cdots\iota(e_{i_s})\varepsilon(e^*_{j_1})\cdots \varepsilon(e_{j_t}^*)\cdot \omega_0.$$ 
Note that $cl(L)_n\cdot \omega_0=0$ for $n>0$ by  (\ref{relation:Fockmodule}), so $\Lambda^{\infty/2+\bullet}L^*$ is a smooth $cl(L)$-module, and the action can be extended to $U_1(cl(L))^{com}:=U(cl(L))^{com}/(K-1)$. 
\medskip 

We want to define an $L$-action on $\Lambda^{\infty/2+\bullet}L^*$ through that of $cl(L)$. 
For $x\in L_n$ with $n\neq 0$, we denote by $\rho(x)$, the action defined by the following
\begin{align}\label{nonzerocomponentaction}
\rho(x)\cdot e_{i_1}^*\wedge e_{i_2}^*\wedge\cdots :=\sum_{k\geq 1}e_{i_1}^*\wedge \cdots \wedge \ad^*x(e_{i_k}^*) \wedge\cdots, 
\end{align}
where $\ad^*$ is the coadjoint action of $L$ on $L^*$. The above sum is finite, thanks to the definition of semi-infinite forms and the fact that $x\in L_n$ for some $n\neq 0$. It is easy to verify the following relations (as operators on $\Lambda^{\infty/2+\bullet}L^*$): for all $y\in L, z^*\in L^*, $
\begin{align}\label{ActonSIF}
[\rho(x), \iota(y)]=\iota(\ad\, x(y)), \qquad [\rho(x), \varepsilon(z^*)]=\varepsilon(\ad^*x(z^*)).
\end{align}

For $x\in L_0$, we cannot use (\ref{nonzerocomponentaction}) as it may involve an infinite sum. Choose $\beta\in L^*_0$ and define $\rho(x)\cdot\omega_0:=\beta(x)\omega_0$ for $x\in L_0$, and then extend to  $\Lambda^{\infty/2+\bullet}L^*$ by requiring (\ref{ActonSIF}).  This can be done because $\Lambda^{\infty/2+\bullet}L^*$ is generated by $\omega_0$ as a  $cl(L)$-module.\medskip 

To give an explicit expression of $\rho(x)$, we define the \emph{normal ordering} :: as follows,
\begin{align*} :\iota(e_i)\iota(e_j)&:=\iota(e_i)\iota(e_j),\,\ :\varepsilon(e_i^*)\varepsilon(e_j^*):=\varepsilon(e_i^*)\varepsilon(e_j^*), \mbox{~for all ~}~ i, j\in \Z, \\ 
-:\varepsilon(e_j^*)\iota(e_i):&=:\iota(e_i)\varepsilon(e_j^*):=\begin{cases} \iota(e_i)\varepsilon(e_j^*)\, &\mbox{~if~} i\neq j \, \mbox{~or~}  i=j\leq 0 ,\\
-\varepsilon(e_j^*)\iota(e_i) \, &\mbox{~if~} i=j>0.
\end{cases} 
\end{align*}

Now for all $x\in L$, the following element is well-defined in $U_1(cl(L))^{com}$,
\begin{align}\label{rho}
\rho^{\beta}(x):=\sum_{i\in \Z}:\iota(\ad\, x(e_i))\varepsilon(e_i^*):+\beta(x).
\end{align} 
The element $\rho^\beta(x)$  acts on $\Lambda^{\infty/2+\bullet}L^*$ since $\Lambda^{\infty/2+\bullet}L^*$ is a smooth $cl(L)$-module.  Moreover, one can show that $\rho^\beta(x)$ satisfies the same relations (\ref{ActonSIF}) as $\rho(x)$ does.  

\begin{lemma}
	The operator $\rho^\beta(x)$ realizes the action of $\rho(x)$ on $\Lambda^{\infty/2+\bullet}L^*$. 
\end{lemma}
\begin{proof}
	It is enough to show  $\rho(x)\cdot \omega_0=\rho^\beta(x)\cdot \omega_0$, since they both satisfy (\ref{ActonSIF}), and  $\omega_0$ generates $\Lambda^{\infty/2+\bullet}L^*$. For simplicity, we assume that $x=e_{i_x}$. By definition \begin{align*}
	\rho(e_{i_x})\cdot \omega_0=\begin{cases}
	\beta(e_{i_x})\omega_0 &\mbox{~if~} e_{i_x}\in L_0, \\
	\sum_{k\geq 0} e_0^*\wedge \cdots \wedge \ad^* e_{i_x}(e_{-k}^*)\wedge \cdots &\mbox{~if~} e_{i_x}\in L_n \mbox{~and~} n\neq 0. 
	\end{cases}
	\end{align*}
	When $e_{i_x}\in L_0$, since $[L_0, L_n]\subseteq L_n$,  there is an annihilation operator in each summand $:\iota(\ad\, e_{i_x}(e_i))\varepsilon(e_i^*):$. Therefore the infinite sum  in  (\ref{rho}) acts as zero on $\omega_0$ and $\rho^\beta(e_{i_x})\cdot \omega_0=\beta(e_{i_x})\omega_0$. When $e_{i_x}\in L_n$ for some $n\neq 0$, we have $\beta(e_{i_x})=0$.  Moreover, we can drop $::$ in (\ref{rho}) as $\varepsilon(\ad^*e_{i_x}(e^*_i))$ always anticommutes with $\iota(e_{i})$ in this case.  Remember that $\iota(e_i)\cdot \omega_0=0$ for all $i> 0$, so we have
	\begin{align*}
	\rho^\beta(e_{i_x})\cdot \omega_0
	&=\sum_{i\leq 0}\varepsilon(\ad^*e_{i_x}(e_i^*))\cdot (-1)^{i}e_0^*\wedge\cdots \wedge \hat{e}^*_{i}\wedge\cdots\\
	&=\sum_{i\leq 0} e_0^*\wedge\cdots \wedge \ad^*e_{i_x}(e_i^*)\wedge\cdots.
	\end{align*}
	
\end{proof}
The centers of the Clifford algebra $Cl(L\oplus L^*)$ and its completion  $U_1(cl(L))^{com}$ are both trivial, i.e., they only contain the constants. For $x, y \in L$, let 
\begin{align}\label{eq:gammadefinition}
\gamma^{\beta}(x, y):=[\rho^{\beta}(x), \rho^{\beta}(y)]-\rho^{\beta}([x, y]).
\end{align}	
It is clear that $\Lambda^{\infty/2+\bullet}L^*$ admits an $L$-module structure under $\rho^\beta(x)$  if and only if $\gamma^\beta\equiv 0$.  One can show that  $\gamma^{\beta}(x, y)$ is central hence a constant in $U_1(cl(L))^{com}$. 
Indeed, it is a 2-cocycle of $L$ \cite{A.Voronov1993}, 
and satisfies  $\gamma^\beta(L_m, L_n)=0$ whenever $m+n\neq 0$.

\begin{definition}
	We say that $L$ admits a {\em semi-infinite structure} through $\rho^{\beta}$ if $\gamma^{\beta}\equiv 0$, i.e., if $\Lambda^{\infty/2+\bullet}L^*$  is an $L$-module under the action $\rho^{\beta}(x)$.  We say that $L$ admits a semi-infinite structure if $\gamma^\beta\equiv 0$ for some $\beta\in L^*_0$.
\end{definition}

\begin{example}
	If $L$ is abelian, it always admits a semi-infinite structure.
	When $H^2(L, \C)=0$, every 2-cocycle is a coboundary. If $\gamma^\beta\neq 0$, we can choose some $\beta'\in L^*$ (by \cite{A.Voronov1993}, we can choose $\beta'\in L^*_0$), such that $\partial \beta'=\gamma^\beta$, then $\rho^{\beta-\beta'}$ gives a semi-infinite structure on $L$.  For example, affine Kac-Moody algebras and the Virasoro algebra admit semi-infinite structures. 
\end{example} 
Let $\mf a$ be a finite-dimensional Lie algebra. Let $\hat{\mf a}:=\mf a\otimes\C[t, t^{-1}]$ be equipped with bracket: $[a\otimes t^n, b\otimes t^m]=[a, b]\otimes t^{n+m}$ for all $a, b\in \mf a$ and $ m, n\in \Z$, where $\C[t, t^{-1}]$ is the ring of Laurent polynomials. It has a natural $\Z$-grading with $\hat{\mf a}_n:=\mf a\otimes t^n$. 
\begin{prop}[\cite{XH}]\label{nilpotentsemistructure}
	Let $\mf n$ be a finite-dimensional nilpotent Lie algebra.  Then $\hat{\mf n}$ admits a semi-infinite structure through $\rho^0$, where $0\in \hat{\mf n}_0$ is the zero function. 
\end{prop}

\subsection{Semi-infinite cohomology} In this subsection,  we assume that $L$ is a quasi-finite $\Z$-graded Lie algebra admitting a semi-infinite structure through $\rho^\beta$ defined by  (\ref{rho}) for some $\beta\in L_0^*$. 
\medskip 

Let $\theta^\beta: L\rightarrow U(L)\otimes U_1(cl(L))^{com}$ be the map defined by  \begin{align}\label{eq:thetadefinition}
\theta^\beta(x):=x+\rho^\beta(x)\quad  (\mbox{~more precisely~} x\otimes 1+1\otimes\rho^\beta(x)),
\end{align} 
which is obviously a Lie algebra homomorphism.  Let $M$ be a smooth $L$-module.  
Then $M\otimes\Lambda^{\infty/2+\bullet} L^*$ is a $U(L)\otimes U_1(cl(L))^{com}$-module hence a smooth $L$-module under the action $\theta^\beta(x)$.  Since $x$ commutes with $\iota(L)$ and $\varepsilon(L^*)$,  for all $y\in L, z^*\in L^*$, we have:
$[\theta^\beta(x), \iota(y)]=\iota([x, y])$  and $[\theta^\beta(x),  \varepsilon(z^*)]=\varepsilon(\ad^*x(z^*)).$

Let 
\begin{align}\label{dd}
d^\beta&=\sum_{i\in \Z}e_i\varepsilon
(e_i^*)-\sum_{i<j}:\iota([e_i, e_j])\varepsilon
(e_i^*)\varepsilon
(e_j^*):+\varepsilon(\beta)\nonumber\\
&=\sum_{i\in \Z}e_i\varepsilon
(e_i^*)-\dfrac{1}{2}\sum_{i, j\in \Z}:\iota([e_i, e_j])\varepsilon
(e_i^*)\varepsilon
(e_j^*):+\varepsilon(\beta).
\end{align} 
Then $d^{\beta}\in U(L)^{com}\otimes U_1(cl(L))^{com}$ has a well-defined action on $M\otimes \Lambda^{\infty/2+\bullet} L^*$.

\begin{lemma}\label{dbeta}
	We have $[d^\beta, \iota(x)]=\theta^\beta(x)$ 	for all $x\in L$. 
\end{lemma}
\begin{proof}
	For simplicity, we assume that $x=e_k$ for some $k\in \Z$. Then 
	\begin{align*}
	\left[\sum_{i\in \Z}e_i\varepsilon
	(e_i^*)+\varepsilon(\beta), \iota(e_k)\right]&=e_k+\beta(e_k), 
	\end{align*}
	and 
	\begin{align*}
	&-\sum_{i<j}\left[:\iota([e_i, e_j])\varepsilon
	(e_i^*)\varepsilon
	(e_j^*):, \iota(e_k)\right]\\
	&\qquad=-\sum_{i<k}:\iota([e_i, e_k])\varepsilon
	(e_i^*):+\sum_{k<j}:\iota([e_k, e_j])\varepsilon
	(e_j^*):\\
	&\qquad=\sum_{i\in \Z}:\iota(\ad\,e_k(e_i))\varepsilon(e_i^*):.
	\end{align*}
	Therefore, we have $[d^\beta, \iota(e_k)]=\theta^\beta(e_k)$. 
\end{proof}

Define a charge grading on $cl(L)$ by setting 
\begin{align}\label{def:chargegradcl}
-\cdeg\, \iota(x)=\cdeg\,\varepsilon(y^*)=1 \quad \mbox{~for~} x\in L, y^*\in L^*, \quad \mbox{~and~ }\quad \cdeg\, K=0 .
\end{align}
When we refer to the charge grading, we will add the superscript $^\divideontimes$. 
The charge grading on $cl(L)$  induces charge gradings on $U(cl(L))$ and  $Cl(L\oplus L^*)$. The space $\Lambda^{\infty/2+\bullet} L^*$ inherits a charge grading if we set $\cdeg\, \omega_0=0$,  and we have $$\Lambda^{\infty/2+n} L^*:=(\Lambda^{\infty/2+\bullet} L^*)_n^\divideontimes=\lspan_\C\{\iota(e_{i_1})\cdots\iota(e_{i_s})\varepsilon(e^*_{j_1})\cdots \varepsilon(e_{j_t}^*)\cdot \omega_0~|~ t-s=n\}.$$
With respect to the charge grading, the operator $\rho^\beta(x)$ is of degree zero for all $x\in L$, so each component $\Lambda^{\infty/2+n} L^*$ is an $L$-submodule. If we define the charge degree of $M$ to be zero, then $d^\beta$ is a charge degree $1$ operator on  $M\otimes \Lambda^{\infty/2+\bullet} L^*$.

\begin{prop}[\cite{A.Voronov1993}, Proposition 2.6]\label{dproperties}	 The operator $d^\beta$ does not depend on the choice of basis of $L$, and  $(d^\beta)^2=0.$ 
\end{prop}

\begin{definition}\label{def:ordinarySIC}
	The complex  $(M\otimes \Lambda^{\infty/2+\bullet} L^*, d^\beta)$ is called the {\em Feigin standard complex} and its cohomology $H^{\infty/2+\bullet}(L, \beta,  M)$ the {\em semi-infinite cohomology} of $L$ with coefficients in $M$. When $\beta=0$, we write just as $H^{\infty/2+\bullet}(L,  M)$.
\end{definition}

\begin{rmk}
	There is an interesting characterization of $d^\beta$ in \cite{FusunAkman} (and in \cite{Arakawa2} for affine W-algebras in the principal nilpotent case). To contrast with our adjusted version in the next section, we call the cohomology in \Cref{def:ordinarySIC} ordinary semi-infinite cohomology. 
\end{rmk}
If $\rho^{\beta'}$ gives another semi-infinite structure on $L$, then we have $(\beta-\beta')([L, L])=0$, and $\beta-\beta'$ defines a 1-dimensional module $\C_{\beta-\beta'}$, on which $x\in L$ acts as $(\beta-\beta')(x)$. 
\begin{prop}[\cite{A.Voronov1993}, Proposition 2.7] If both $\rho^\beta$ and $\rho^{\beta'}$ give semi-infinite structures on $L$, then 
	$$H^{\infty/2+\bullet}(L, \beta, M) \cong H^{\infty/2+\bullet}(L, \beta', M\otimes \C_{\beta-\beta'}).$$
\end{prop}

\medskip
\section{An adjustment when the 2-cocycle $\gamma^\beta$ is not identically zero}\label{sec:2}

Recall the notation in the previous section. We assume that $\gamma^\beta$ is not identically zero in this section,  i.e.,  $\rho^\beta$ does not give a semi-infinite structure on $L$. 

\subsection{What is the problem }
Let $d^\beta$ be defined by (\ref{dd}) and consider the value $[[(d^\beta)^2, \iota(x)], \iota(y)]$ for $x, y\in L$.  Since $d^\beta$ is odd, we have  $(d^\beta)^2=\dfrac{1}{2}[d^\beta, d^\beta]$ hence $[(d^\beta)^2, \iota(x)]=[d^\beta, [d^\beta, \iota(x)]]$. By \Cref{dbeta}, we have  $[d^\beta, \iota(x)]=\theta^\beta(x)$ (though we assume $\gamma^\beta\equiv 0$ in that section, the calculations there still hold), so 
\begin{align}\label{eq:whynotzero}
[[(d^\beta)^2, \iota(x)], \iota(y)]&=
[[d^{\beta}, \theta^{\beta}(x)], \iota(y)]\nonumber\\
&=[d^{\beta}, [\theta^{\beta}(x), \iota(y)]]+[[d^{\beta}, \iota(y)], \theta^{\beta}(x)]\nonumber\\
&=[d^{\beta}, \iota([x, y])]+[\theta^{\beta}(y), \theta^{\beta}(x)]\nonumber\\
&=\theta^{\beta}([x, y])-[\theta^{\beta}(x), \theta^{\beta}(y)].
\end{align}
Note that $\theta^{\beta}([x, y])-[\theta^{\beta}(x), \theta^{\beta}(y)]=-\gamma^{\beta}(x, y)$. In particular, the operator $d^\beta$ is not square zero if $\gamma^{\beta}$ is not identically zero. This is the problem!\medskip

\subsection{Construction of new ``ghosts''.}
Let $\ker\, \gamma^\beta:=\{x\in L~|~ \gamma(x, L)\equiv 0\}$ be the radical of $\gamma^\beta$, which is obviously a graded subalgebra of $L$. Let $F_\beta$ be a graded complement of $\ker\, \gamma^\beta$ in $L$. 
Let $\epsilon(F_\beta)$ be a copy of $F_\beta$. For $x\in L$, let $\epsilon(x)$ be the projection in $F_\beta$ but considered as an element of $\epsilon(F_\beta)$. Note that $\epsilon(\ker\,\gamma^\beta)=0$. Consider the following Lie superalgebra (which contains $cl(L)$ as a subalgebra,) $$c(L):=\iota(L)\oplus \varepsilon(L^*)\oplus \C K\oplus \epsilon(F_\beta),$$ where  $\epsilon(F_\beta)$ is defined to be even, it commutes with $cl(L)$ and has bracket: for $x, y\in F_\beta$, $[\epsilon(x), \epsilon(y)]:=-\gamma^\beta(x, y)K$.  
The subalgebra $\epsilon(F_\beta)\oplus \C K$ is $\Z$-graded.  
Its abelian subalgebra $\epsilon(F_\beta)_+:=\left(\bigoplus_{n>0}\epsilon(F_\beta)_n\right)\oplus \C K$ has a  1-dimensional module $\C$, on which $\bigoplus_{n>0}\epsilon(F_\beta)_n$ acts as zero and $K$ acts  identity. The  {\em Fock representation} of $\epsilon(F_\beta)\oplus \C K$ is defined to be the induced (and obviously smooth) module
\begin{equation}\label{eq:fockrepres}
\mf F_{\beta}=\Ind_{\epsilon(F_\beta)_+}^{\epsilon(F_\beta)\oplus \C K}\C.
\end{equation}
Remember that $\Lambda^{\infty/2+\bullet} L^*$ is a smooth $cl(L)$-module on which $K$ also acts as identity, so  $\Lambda^{\infty/2+\bullet}L^*\otimes \mf F_\beta$ is a smooth $c(L)$-module. 
\medskip 

Let $s(L)=L\oplus c(L)$ be the direct sum of $L$ and $c(L)$. For $x\in L$, let 
\begin{align*}
\bar{\rho}^\beta(x):=\rho^\beta(x)+\epsilon(x)\in U_1(c(L))^{com}:=U(c(L))^{com}/(K-1),
\end{align*}
and 
\begin{align}\label{thetabar}
\bar{\theta}^\beta(x)=x+\bar{\rho}^\beta(x)\in U_1(s(L))^{com}:=U(s(L))^{com}/(K-1). 
\end{align}
Then $\bar{\rho}^\beta(x)$ has a well-defined action on $\Lambda^{\infty/2+\bullet}L^*\otimes \mf F_\beta$. Let $M$ be a smooth $L$-module. Then $\bar{\theta}^\beta(x)$ has a well-defined action on $M\otimes \Lambda^{\infty/2+\bullet}L^*\otimes \mf F_\beta$. Moreover, both $\bar{\rho}^\beta(x)$ and $\bar{\theta}^\beta(x)$ satisfy the relations (\ref{ActonSIF}),  and for all $x, y\in L$,  we have
\begin{align}\label{rhobar}
[\bar{\rho}^\beta(x), \epsilon(y)]=[\bar{\theta}^\beta(x), \epsilon(y)]=-\gamma^\beta(x, y).
\end{align}

\begin{lemma}\label{liealghom}
	The map $\bar{\rho}^\beta: L \longrightarrow U_1(c(L))^{com}$ sending $x$ to $\bar{\rho}^\beta(x)$ is a Lie algebra homomorphism if $[L, L]\subseteq \ker\,\gamma^\beta$. 
\end{lemma}
\begin{proof}
	We need to prove  $\bar{\rho}^\beta([x, y])=[\bar{\rho}^\beta(x), \bar{\rho}^\beta(y)]$ for all $x, y \in L$. But we have 
	\begin{align*}
	[\bar{\rho}^\beta(x), \bar{\rho}^\beta(y)]&=[\rho^\beta(x)+\epsilon(x), \rho^\beta(y)+\epsilon(y)]\\
	&=[\rho^\beta(x), \rho^\beta(y)]+[\epsilon(x), \epsilon(y)]\\
	&=\rho^\beta([x, y])+\gamma^\beta(x, y)-\gamma^\beta(x, y)\\
	&=\rho^\beta([x, y]),
	\end{align*}
	and $\bar{\rho}^\beta([x, y])=\rho^\beta([x, y])$ if $\epsilon([x, y])\equiv 0$, i.e., if $[L, L]\subseteq \ker\,\gamma^\beta$. 
\end{proof}

\medskip
\noindent\textbf{Assumption:} From now on,  we assume that $[L, L]\subseteq \ker\,\gamma^\beta$ is satisfied.
\medskip

\begin{rmk}\label{rmkbeforetheassumption}
	\Cref{liealghom} tells us that the tensor product $\Lambda^{\infty/2+\bullet} L^* \otimes \mf F_\beta$ is an $L$-module under  $\bar{\rho}^\beta(x)$  though $\Lambda^{\infty/2+\bullet} L^*$ is not under the action $\rho^\beta(x)$. This Fock module $\mf F_\beta$ will be the new ``ghosts'' to be added!
\end{rmk}

\subsection{Construction and characterization of a square zero differential}
Extend the charge grading (see (\ref{def:chargegradcl})) on $cl(L)$ to $c(L)$ and $s(L)$ by setting $\cdeg~ \epsilon(F_\beta)=\cdeg\, L=0$. 
At the module level, set $\cdeg\, M=\cdeg\, \mf F_\beta=0$, then  $\Lambda^{\infty/2+\bullet} L^*\otimes \mf F_{\beta}$ is a $\Z$-graded $c(L)$-module and  $M\otimes \Lambda^{\infty/2+\bullet} L^*\otimes \mf F_{\beta}$  a $\Z$-graded $s(L)$-module with respect to the charge gradings.  Let $i_c: c(L)\hookrightarrow U_1(c(L))^{com}$ and $i_s: s(L) \hookrightarrow U_1(s(L))^{com}$ be the canonical inclusions. 
\begin{definition}
	A superderivation $D$ with respect to $i_c$ or $i_s$,  is said to be of {\em charge degree} $N$ if $D(c(L)^\divideontimes_n)\subseteq U_1(c(L))^{com, \divideontimes}_{n+N}$ or $D(s(L)^\divideontimes_n)\subseteq U_1(s(L))^{com, \divideontimes}_{n+N}$, respectively. A superderivation $D$ of $c(L)$ or of $s(L)$ is said to be of {\em charge degree} $N$ if $D(c(L)^\divideontimes_n)\subseteq c(L)^\divideontimes_{n+N}$ or $D(s(L)^\divideontimes_n)\subseteq s(L)^\divideontimes_{n+N}$, respectively.
\end{definition}

Define an action of $L$ on $c(L)$ as follows. For $x, y\in L, z\in L^*$, 
$$ x\cdot \iota(y)=\iota([x, y]), \quad x\cdot \varepsilon(z^*)=\varepsilon(\ad^*x(z^*)),\quad x\cdot \epsilon(y)=-\gamma^\beta(x, y)K,\quad x\cdot K=0.$$ 
We extend this action to $s(L)$ by letting $L$ act on itself by the adjoint action. 

\begin{rmk} The actions of $x\in L$ on $c(L)$ and $s(L)$ defined above are even derivations of charge degree zero.
	They induce even derivations of charge degree zero  on $U_1(c(L))^{com}$ and $U_1(s(L))^{com}$, respectively. The inner derivations $[\bar{\rho}^\beta(x), \cdot]$ and $[\bar{\theta}^\beta(x), \cdot]$ realize these actions, respectively, by (\ref{ActonSIF}) and (\ref{rhobar}). 
\end{rmk}
\begin{lemma}\label{chargedegreenegative}
	Let $u\in U_1(s(L))^{com}$ be a charge degree $\geq 1$ element. Then $[u, \iota(x)]=0$ for all $x\in L$ only if $u=0$.  
\end{lemma}
\begin{proof}
	As $\cdeg\, u\geq 1$, if $u$ is not zero, we can write $$u=w\varepsilon(e_k^*)+v\qquad \mbox{~or~}\qquad u=\varepsilon(e_k^*)w+v$$ for some $k\in \Z$ with $w, v\in U_1(s(L))^{com}$ and $w\neq 0$, such that $\varepsilon(e_{k}^*)$ does not appear in $w$ or $v$, i.e.,  $$[w, \iota(e_k)]=[v, \iota(e_k)]=0. $$
	Then $[u, \iota(e_k)]=w\neq 0$  gives a contradiction.  
\end{proof}

\begin{lemma}\label{uniquederivation}  Let $D$ be a superderivation with respect to $i_s$ of charge degree $\geq 1$,  and suppose that $D(K)=0$. Then $D$ is determined by its value on $\iota(L)$. 
\end{lemma}
\begin{proof} 
	Let $D'$ be another superderivation of same parity, such that $D'(K)=0$ and coincide with $D$ on $\iota(L)$. We show that $D=D'$. For all $u, v\in s(L),$ we have
	\begin{align}\label{eq:d1minusd2}
	(D-D')[u , v]=[(D-D')u, v]+(-1)^{i\cdot p(u)} [u,   (D-D')v],
	\end{align}
	where $i$ is the parity of $D$ and $D'$. Note that $[s(L), \iota(L)]\subseteq \C K$ and $(D-D')K=(D-D')\iota(L)=0$. Let $u\in s(L), v=\iota(x)\in \iota(L)$ in (\ref{eq:d1minusd2}). Then we have 	\begin{align}\label{eq:oneequation}
	[(D-D')u, \iota(x)]=0.
	\end{align}
	If $u\in \iota(L)$, then $(D-D')u=0$. If $u\in L\oplus \epsilon(F_\beta)\oplus\varepsilon(L^*)$, then note that $\cdeg\, (D-D')u\geq 1$ if it is not zero. Since (\ref{eq:oneequation}) holds for all $\iota(x)\in \iota(L)$,  \Cref{chargedegreenegative} ensures that $(D-D')u=0$, i.e., $D=D'$ on $s(L)$. 
\end{proof}
\begin{remark} Given a charge degree $\geq 1$ superderivation with respect to the inclusion $i_{\iota(L)}: \iota(L)\rightarrow U_1(s(L))^{com}$, \Cref{uniquederivation} says that we can extend it to a superderivation of the same charge degree with respect to $i_s$ in a unique way. 
\end{remark}
Recall that $\bar{\theta}^\beta(x)$ defined by (\ref{thetabar}) is even and satisfies (\ref{ActonSIF}), in particular,
$$[\bar{\theta}^\beta(x), \iota(y)]-[\iota(x), \bar{\theta}^\beta(y)]=\iota([x, y])+\iota([y, x])=0.$$ 
As $\iota(L)$ is an abelian subalgebra of $s(L)$, the map $D: \iota(L)\rightarrow U_1(s(L))^{com}$ sending $\iota(x)$ to $\bar{\theta}^\beta(x)$ is an odd superderivation of charge degree $1$ with respect to $i_{\iota(L)}$, so it can be extended to be a superderivation with respect to $i_s$ in a unique way.

Recall the expression $d^\beta$ defined by (\ref{dd}), let \begin{align}\label{def:dbar}
\bar{d}^\beta=d^\beta+\sum_{i\in \Z}\varepsilon(e_i^*)\epsilon(e_i),
\end{align}

\begin{theorem}\label{maintheorem}
	We  have $(\bar{d}^\beta)^2=0$, and 
	it is the unique element  of charge degree $1$ in $U_1(s(L))^{com}$ satisfying $[\bar{d}^\beta, \iota(x)]=\bar{\theta}^\beta(x) \, \mbox{~for all ~} x\in L$. 
\end{theorem}
\begin{proof}
	By \Cref{dbeta}, we have $[d^\beta, \iota(x)]=\theta^\beta(x)$, so we only need to show that 
	\begin{align*}
	\sum_{i\in \Z}[\varepsilon(e_i^*) \epsilon(e_i), \iota(x)]=\epsilon(x).
	\end{align*}
	This is obvious for $x=e_k$ hence true for all $x\in L$. 
	The uniqueness is  by \Cref{chargedegreenegative}.

	The inner derivations $[(\bar{d}^\beta)^2, \cdot]$ and $[[(\bar{d}^\beta)^2, \iota(x)], \cdot]$ are  of  charge degree $2$ and degree $1$, respectively, if they are non-zero. By \Cref{uniquederivation}, they are completely determined by their value on $\iota(L)$. Recall the calculations in  (\ref{eq:whynotzero}). Since $[L, L]\subseteq \ker\,\gamma^\beta$ and  $[\bar{d}^\beta, \iota(x)]=\bar{\theta}^\beta(x)$, for $x, y\in L$, we have  
	\begin{align*}
	[[(\bar{d}^\beta)^2, \iota(x)], \iota(y)]
	&=\bar{\theta}^\beta([x, y])-[\bar{\theta}^\beta(x), \bar{\theta}^\beta(y)]\\
	&=\rho^\beta([x, y])+[x, y]-[\rho^\beta(x)+x+\epsilon(x), \rho^\beta(y)+y+\epsilon(y)]\\
	&=\rho^\beta([x, y])-[\rho^\beta(x), \rho^\beta(y)]+\gamma^\beta(x, y)\\
	&=0.
	\end{align*}
	\Cref{chargedegreenegative} then implies that  $[(\bar{d}^\beta)^2, \iota(x)]=0$ for all $x\in L$ hence $(\bar{d}^\beta)^2=0$.
\end{proof}

\begin{definition}
	We call the complex $(M\otimes \Lambda^{\infty/2+\bullet} L^*\otimes \mf F_{\beta}, \bar{d}^\beta)$ the {\em adjusted Feigin complex} with respect to $\beta$, and its cohomology $H_a^{\infty/2+\bullet}(L, \beta, M)$ the {\em adjusted semi-infinite cohomology} of $L$ with coefficients in $M$, with respect to $\beta$. 
\end{definition}
\begin{remark} 
	Note that we used a subscript ``$a$'' in the adjusted semi-infinite cohomology in contrast to ordinary semi-infinite cohomology. 	
\end{remark}

\subsection{Comparison with ordinary semi-infinite cohomology}
The adjustment sometimes gives nothing new but ordinary semi-infinite cohomology with coefficients in another module. Assume that $\rho^\beta$ gives a semi-infinite structure on $L$, and $\beta'\in \bigoplus_{n\geq 0}L_n^*$ 
such that $\partial \beta'\neq 0$ but $\partial \beta'([L, L], L)=0$, where $\partial \beta'(x, y)=\beta'([x, y])$. Then $\gamma^{\beta+\beta'}=-\partial \beta'\neq 0$ and $[L, L]\subseteq \ker\, \gamma^{\beta+\beta'}$.  We can therefore talk about the adjusted semi-infinite cohomology of $L$ with coefficients in a smooth module $M$,  with respect to $\beta+\beta'$, i.e., the cohomology of $(M\otimes\Lambda^{\infty/2+\bullet }\otimes \mf F_{\beta+\beta'}, \bar{d}^{\beta+\beta'})$.  Recall that \begin{align*}
\bar{d}^{\beta+\beta'}&=\sum_{i\in \Z}e_i\varepsilon
(e_i^*)-\dfrac{1}{2}\sum_{i, j\in \Z}:\iota([e_i, e_j])\varepsilon
(e_i^*)\varepsilon
(e_j^*):+\varepsilon(\beta+\beta')+\sum_{i\in \Z}\varepsilon(e_i^*) \epsilon(e_i)\nonumber\\
&=\sum_{i\in \Z}\varepsilon(e_i^*) (e_i+\beta'(e_i)+\epsilon(e_i))-\dfrac{1}{2}\sum_{i, j\in \Z}:\iota([e_i, e_j])\varepsilon
(e_i^*)\varepsilon
(e_j^*):+\varepsilon(\beta),
\end{align*}
and
\begin{align*}
[\bar{d}^{\beta+\beta'}, \iota(x)]=x+\beta'(x)+\epsilon(x)+\rho^\beta(x). 
\end{align*}
Since $[\epsilon(x), \epsilon(y)]=-\gamma^{\beta+\beta'}(x, y)=\beta'([x, y])$ and $\epsilon([x, y])\equiv 0$, we have
\begin{align*}
[x+\beta'(x)+\epsilon(x), y+\beta'(y)+\epsilon(y)]=[x, y]+\beta'([x, y]), 
\end{align*}
that is, $M\otimes \mf F_{\beta+\beta'}$ becomes an $L$-module under the action $x+\beta'(x)+\epsilon(x)$, and it is smooth.  Therefore, we have the following theorem. 
\begin{theorem}\label{prop:adjustedandordinary}
	Let $\beta, \beta'$ be as above. Then $$H_a^{\infty/2+\bullet}(L, \beta+\beta', M) \cong H^{\infty/2+\bullet}(L, \beta, M\otimes \mf F_{\beta+\beta'}).$$
\end{theorem}

\medskip
\section{Affine W-algebras}\label{sec:3}
In the mathematical literature, affine W-algebras were introduced in the principal nilpotent case \cite{KacFrenkelWakimoto} about ten years earlier than in the general case \cite{Kac&Wakimoto2003}. One reason is that in the general case, we need to add more ``ghosts'' \cite{Bershadsky}. 
In this section, we will explain the definition of  affine $W$-algebras  associated to good $\Z$-gradings \cite{Kac&Wakimoto2003} in the language of adjusted semi-infinite cohomology and give them a uniform realization. We will also give a characterization of admissible pairs introduced in \cite{Sadaka} and define affine W-algebras associated to some special admissible pairs. 
\medskip 

\subsection{Affine W-algebras associated to good $\Z$-gradings}
Let $\mf g$ be a finite-dimensional semisimple Lie algebra with Killing form $(\cdot | \cdot)$, and $e\in \mf g$  a nonzero nilpotent element. Let $\Gamma: \mf g=\bigoplus_{i\in \Z}\mf g_i$ be a $\Z$-grading of $\mf g$. 
\begin{definition}
	The $\Z$-grading $\Gamma$ is called a \emph{good $\Z$-grading} with respect to $e$, if   $e\in \mf g_2$ and $\ad\,e:  \mf g_i\rightarrow \mf g_{i+2}$ is injective for $i\leq -1$ and surjective for $i\geq -1$.
\end{definition} For example, by Jacobson-Malecov's theorem, we can embed $e$ into an $\mf sl_2$-triple $\{e, f, h\}$ in $\mf g$.  
Then  $\mf g=\bigoplus_{i\in \Z}\mf g_i, \, \mbox{~where~} \mf g_i:=\{y\in \mf g~|~ [h, y]=iy\}$,  is a good $\Z$-grading with respect to $e$. Such $\Z$-gradings are called \emph{Dynkin gradings}. The $\Z$-grading  
$\Gamma$ is called an \emph{even grading} if $\mf g_i=0$ for all odd $i$.

There is a non-degenerate bilinear form on $\mf g_{-1}$ defined by  $\langle a, b\rangle:=(e~|~[b, a])$ for $a, b\in \mf g_{-1}$. Let $\mf l$ be an isotropic subspace of $\mf g_{-1}$, i.e., $\langle x, y\rangle= 0$ for all $x, y\in \mf l$,  and $\mf l^\perp:=\{x\in \mf g_{-1}~|~\langle x, \mf l\rangle=0 \}$ its orthogonal complement. We have $\mf l\subseteq \mf l^\perp$ and the bilinear form $\langle \cdot, \cdot \rangle$ is non-degenerate when restricted to $\mf l^\perp/ \mf l$. Let $\mf m=\bigoplus_{i\leq -2}\mf g_i$,    $\mf m_{\mf l}=\mf m\oplus \mf l$,  $\mf n_{\mf l}=\mf m\oplus \mf l^\perp$, and $\mf n=\mf m\oplus \mf g_{-1}$, which are all nilpotent subalgebras of $\mf g$. 

The {\em Kac-Moody affinization} of $\mf g$ is $\hat{\mf g}=\left(\mf g\otimes \C[t, t^{-1}]\right)\oplus \C K$  with Lie bracket: 
\begin{align*} 
[a\otimes t^m, b\otimes t^n]=[a, b]\otimes t^{m+n}+m\delta_{m, -n}(a~|~b)K,  \quad [K, \hat{\mf g}]=0. 
\end{align*}
Let $\hat{\mf g}_+=\left(\mf g\otimes \C [t]\right)\oplus \C K$, which is a subalgebra of $\hat{\mf g}$.  Define a 1-dimensional module $\C_{k}$ of $\hat{\mf g}_+$ on which $\mf g\otimes \C [t]$ acts as zero and $K$ acts as the constant $k$.  The \emph{vacuum representation} of level $k$ is the induced  module \begin{align*}\label{def:vacuum}
V_k(\mf g):=\Ind_{\hat{\mf g}_+}^{\hat{\mf g}} \, \C_{k}.
\end{align*}
It is a smooth $\hat{\mf g}$-module, and there is a vertex algebra structure on $V_k(\mf g)$.

Choose a basis $\{u_\alpha\}_{\alpha\in S_j}$ of each $\mf g_j$, and a dual basis $\{u_\alpha^*\}_{\alpha\in S_j}$ for $\mf g_j^*$. Let $S_{\mf l^\perp }\subseteq S_{-1}$ such that $\{u_\alpha\}_{\alpha\in S_{\mf l^\perp }}$ is a basis of $\mf l^\perp$. Let $S_{\mf n_{\mf l}}=\bigcup_{j<-1}S_j\cup S_{\mf l^\perp}$.
Let  $\hat{\mf n}_{\mf l}=\mf n_{\mf l} \otimes \C[t, t^{-1}]$  be the affinization of $\mf n_{\mf l}$, and $\hat{\mf n}_{\mf l}^*:=\mf n_{\mf l}^*\otimes \C[t, t^{-1}]$. 
Denote by $u_{i, n}=u_i\otimes t^n$ and $u_{i, n}^{*}=u_i^*\otimes t^{n}$. Then $\{u_{i, n} \}_{n\in \Z, i\in S_{\mf n_{\mf l}}}$ and $\{u_{i, n}^{*} \}_{n\in \Z, i\in S_{\mf n_{\mf l}}}$ form  bases of $\hat{\mf n}_{\mf l}$ and $\hat{\mf n}^*_{\mf l}$, respectively. One can identify $\hat{\mf n}^*_{\mf l}$ with the restricted dual of $\hat{\mf n}_{\mf l}$ under the paring $\langle u_{i, n}, u_{j, m}^{*}\rangle:=\delta_{m, -n-1}\delta_{i, j}$. Note that we have a shift of index. 
As $\mf n_{\mf l}$ is nilpotent, it admits a semi-infinite structure.  Let $\beta_e\in \hat{\mf n}^*_{\mf l}$ be defined by $\beta_e(u\otimes t^n):=\delta_{n,-1}(e~|~u)$ for $u\in \mf n_{\mf l}$.

Let
\begin{align}
\rho^{\beta_e}(x)=\sum_{i\in S_{\mf n_{\mf l}}, n\in \Z}:\iota(\ad\, x(u_{i, n}))\varepsilon(u_{i, -n-1}^{*}):+\beta_e(x).
\end{align}  Then for $x, y\in \hat{\mf n}_{\mf l}$, $$\gamma^{\beta_e}(x, y):=[\rho^{\beta_e}(x), \rho^{\beta_e}(y)]-\rho^{\beta_e}([x, y])=-{\beta_e}([x, y]).$$ 
In particular, $\rho^{\beta_e}(x)$ gives a semi-infinite structure on $\hat{\mf n}_{\mf l}$ if and only if $\beta_e([x, y])= 0$ for all $x, y\in \mf n_{\mf l}$, which is true if and only if the $\Z$-grading  $\Gamma$ is even. 

Let $\hat{\mf m}_{\mf l}=\mf m_{\mf l}\otimes \C[t, t^{-1}]$. Note that $\ker\,\gamma^{\beta_e}=\hat{\mf m}_{\mf l}$. Moreover, we have $[\hat{\mf n}_{\mf l}, \hat{\mf n}_{\mf l}]\subseteq \hat{\mf m}_{\mf l}$, hence the assumption after \Cref{liealghom} with respect to the 1-form $\beta_e$ is satisfied for $\hat{\mf n}_{\mf l}$. We can consider the adjusted semi-infinite cohomology of $\hat{\mf n}_{\mf l}$ with coefficients in the smooth module $V_k(\mf g)$, with respect to $\beta_e$.
\begin{definition}[\cite{Kac&Wakimoto2003}]
	The affine W-algebra $W^k(\mf g, e)$ associated to the data $(\mf g, e, k)$ is the adjusted semi-infinite cohomology  $H_a^{\infty/2+\bullet}(\hat{\mf n}_{\mf l}, \beta_e, V_k(\mf g))$. 
\end{definition}
\begin{remark} In \cite{Kac&Wakimoto2003}, the authors set $\mf l=0$. In that case, the adjusted Feigin complex is $V_k(\mf g)\otimes \Lambda^{\infty/2+\bullet}\hat{\mf n}\otimes \mf F_{\beta_e}$. One can show that $\Lambda^{\infty/2+\bullet}\hat{\mf n}$  and $\mf F_{\beta_e}$ correspond to $F(A_{ch})$ and   to $F(A_{ne})$, respectively,  as in \cite{Kac&Wakimoto2003}, hence the complexes coincide. 
	A detailed calculation showing that the differentials also coincide and hence the two  definitions are equivalent can be found in \cite{XH}.
\end{remark}

By \Cref{prop:adjustedandordinary}, we have the following uniform definition of affine W-algebras.
\begin{theorem}\label{SecondMainTheorem}  Let $\mf g$ be a  finite-dimensional semisimple Lie algebra and $e\in \mf g$  a general non-zero nilpotent element. Then
	$$W^k(\mf g, e)\cong H^{\infty/2+\bullet}(\hat{\mf n}_{\mf l}, V_k(\mf g)\otimes\mf F_{\beta_e}).$$
\end{theorem}

\begin{remark}\leavevmode
	\begin{itemize}
		\item[(1)] When the $\Z$-grading  is even, i.e., $\mf n=\mf m$,  $\rho^{\beta_e}$ gives a semi-infinite structure on $\mf n$, and  the Fock module $\mf F_{\beta_e}$ reduces to a $1$-dimensional module on which $x\in \hat{\mf n}$ acts as $\beta_e(x)$. This recovers the semi-infinite cohomology realization of affine W-algebra in the principal nilpotent case \cite{KacFrenkelWakimoto}. 
	\end{itemize}
\begin{itemize}
\item[(2)]
		The realization of $W^k(\mf g, e)$ through $ H^{\infty/2+\bullet}(\hat{\mf n}_{\mf l}, V_k(\mf g)\otimes\mf F_{\beta_e})$ was also observed (Remark 3.6.1) in \cite{Arakawa3}, though the construction there is a bit different from ours. 
	\end{itemize}
\end{remark}

\subsection{Characterization of an admissible pair}
In the paper \cite{Sadaka}, the author introduced the notion of an admissible pair with respect to a nilpotent element and defined finite W-algebras associated to admissible pairs. We are going to give a characterization of admissible pairs and define affine W-algebras associated to some special admissible pairs using adjustedMaciej Zakarczemny semi-infinite cohomology. 

Let $\mf g$ be a finite-dimensional semisimple Lie algebra with Killing form $(\cdot | \cdot)$, and $e\in \mf g$  a nonzero nilpotent element. Let $\Gamma: \mf g=\bigoplus_{i\in \Z}\mf g_i$ be a $\Z$-grading of $\mf g$. 
\begin{definition}\label{Def:AdPare}
A pair $(\mf m, \mf n)$  of graded  subalgebras (with respect to $\Gamma$) is called an \emph{admissible pair} with respect to $e$ if there exists an integer $a>1$, such that 
	\begin{itemize}
		\item[(i)] $e\in \mf g_a$; 
		\item[(ii)] $\bigoplus_{i\leq -a}\mf g_i\subseteq \mf m\subseteq \mf n\subseteq \bigoplus_{i<0}\mf g_i$; 
		\item[(iii)] $\mf m^\perp \cap [\mf g, e]=[\mf n, e]$, where $\mf m^\perp:=\{x\in \mf g~|~ (\mf m~|~ x)=0 \}$;
		\item[(iv)] $\ad\, e: \mf n\rightarrow [\mf n, e]$ is injective;
		\item[(v)] $[\mf m, \mf n]\subseteq \mf m$;
		\item[(vi)] $\dim\,\mf m+\dim\,\mf n=\dim\, [\mf g, e]$. 
	\end{itemize}
\end{definition}
Given a good $\Z$-grading $\mf g=\bigoplus_{i\in \Z}\mf g_i$ with respect to $e$, and an isotropic subspace $\mf l$ of $\mf g_{-1}$, the pair $(\mf m_{\mf l}, \mf n_{\mf l})$ defined as in the previous section is an admissible pair with $a=2$. In \cite{Sadaka}, some admissible pairs were shown not to be induced from a good $\Z$-grading. So an admissible pair is a more general notion. Finite W-algebras associated to admissible pairs were introduced in \cite{Sadaka}, and they were proved to be isomorphic to finite W-algebras associated to good $\Z$-gradings in some cases. It was also conjectured that this is true in general \cite{Sadaka}. 

Here we give another characterization of admissible pairs and introduce affine W-algebras associated to some special admissible pairs. 

Assume that $(\mf m, \mf n)$ is an admissible pair with respect to $e$ and the $\Z$-grading $\Gamma: \mf g=\bigoplus_{i\in \Z}\mf g_i$, where $e\in \mf g_a$ for some integer $a>1$. Let $$\mf g_{\leq -a}:=\bigoplus_{i\leq -a}\mf g_i, \quad \mf g_{<0}:=\bigoplus_{i<0}\mf g_i, \quad \mf g_{(-a,0)}:=\bigoplus_{-a<i<0}\mf g_i, \quad \mf g_{(0,a)}:=\bigoplus_{0<i<a}\mf g_i.$$ Let $\mf g_{<0}^e:=\{x\in \mf g_{<0}~|~[e, x]=0\}$. Then $\mf g_{<0}^e$ is a graded subspace of $\mf g_{(-a,0)}$ by Condition (iv) in \Cref{Def:AdPare}. Let $\mf n_{>-a}=\mf n\cap \mf g_{(-a,0)}$ and $\mf m_{>-a}=\mf m\cap \mf g_{(-a,0)}$.  
Since $\mf n_{>-a}\cap \mf g^e_{<0}=\{0\}$, one can extend $\mf n_{>-a}$ to a graded complement of $\mf g_{<0}^e$ in $\mf g_{(-a,0)}$, which we denote by $\mf g^{>-a}$ (this notation might be confusing, so we use superscript instead of subscript). Note that  we have $\mf g_{(-a, 0)}=\mf g_{<0}^e\oplus \mf g^{>-a}$ and $\mf m_{>-a}\subseteq \mf n_{>-a}\subseteq \mf g^{>-a}$. Obviously,  $\ad\,e$ is injective on $\mf g^{>-a}$. Let $ \mf g^{>-a}_i= \mf g^{>-a}\cap \mf g_i$.

Define a skew-symmetric bilinear form on $\mf g^{>-a}$ by 
\begin{align}\label{Def:BFong}
\langle x, y\rangle :=(e~|~[y, x]).
\end{align}
\begin{lemma}\label{lem:nondegenerate}
	The bilinear form on $\mf g^{>-a}$  defined by (\ref{Def:BFong}) is non-degenerate, and $\mf g^{>-a}$ is symmetric with respect to $-a/2$ in the sense that there is a non-degenerate pairing between $\mf g^{>-a}_{-a/2+i}$ and $\mf g^{>-a}_{-a/2-i}$, for each $i\in a/2+ \Z$ and $-a/2<i<a/2$.    
\end{lemma}
\begin{proof}
	Since $(\cdot~|~\cdot)$ is non-degenerate and invariant on $\mf g$, it restricts to a non-degenerate pairing between $\mf g_{(-a, 0)}$	and $\mf g_{(0, a)}$. Given $x\in \mf g^{>-a}$, since $\ad\, e$ is injective on $\mf g^{>-a}$,  $[e, x]\in \mf g_{(0, a)}$ is nonzero, so there exists some $y'\in \mf g_{(-a, 0)}$ such that $([e, x]~|~y')\neq 0$. Let $y'=y+z$ with $z\in \mf g^e_{<0}$ and $y\in \mf g^{>-a}$, then since $([e, x]~|~z)=(x~|~[z, e])=0$, we have $\langle y, x \rangle=([e, x]~|~y)\neq 0$. 
	
	We now check the symmetry with respect to $-a/2$. Let $x\in \mf g^{>-a}_{-a/2-i}$, and assume that $y\in \mf  g^{>-a}$ such that $([e, x]~|~y)\neq 0$. Let $y=\sum_i y_i$ with $y_i\in \mf g_i$. Then since $[e, x]\in \mf g_{a/2-i}$, we have $([e, x]~|~y_i)=0$ for all $i\neq -a/2+i$ and  hence $([e, x]~|~y_{-a/2+i})\neq 0$. As $\mf g^{>-a}$ is a graded subspace, we have $y_{-a/2+i}\in \mf g^{>-a}$.

\end{proof}
\begin{lemma}\label{Cor:dimequality}
	We have $\dim \mf g^{>-a}+2\dim\,\mf g_{\leq -a}=\dim [\mf g, e]$. 
\end{lemma}
\begin{proof} Note that $[\mf g, e]$ is a graded subspace of $\mf g$. Let $\mf g_{\geq a}:=\bigoplus_{i\geq a}\mf g_i$.  Denote by $[\mf g, e]_{\leq 0}=[\mf g, e]\cap \mf g_{\leq 0}$,  $[\mf g, e]_{(0, a)}=[\mf g, e]\cap\mf g_{(0, a)}$ and $[\mf g, e]_{\geq a}=[\mf g, e]\cap \mf g_{\geq a}$. It was proved in \cite{Sadaka} that $\mf g_{\geq a}\subseteq [\mf g, e]$, so $\dim\, [\mf g, e]_{\geq a}=\dim\, \mf g_{\geq a}$. Since $\ad\, e$ is injective on $\mf g_{\leq -a}$ and $[\mf g, e]_{\leq 0}=[\bigoplus_{i\leq -a}\mf g_i, e]$, we have $\dim \,[\mf g, e]_{\leq 0}=\dim \,\mf g_{\leq -a}.$ Finally, we have $$[\mf g, e]_{(0, a)}=[\mf g_{(-a, 0)}, e]=[\mf g_{<0}^e\oplus \mf g^{>-a}, e]=[\mf g^{>-a}, e].$$ Since $\ad\, e$ is injective on $\mf g^{>-a}$, we have $\dim [\mf g, e]_{(0, a)}=\dim\,\mf g^{>-a}$. The non-degenerate pairing between $\mf g_{\leq -a}$ and $\mf g_{\geq a}$ ensures that $\dim\,\mf g_{\leq -a}=\dim\,\mf g_{\geq a}$. Now the equality of dimensions is clear. 
\end{proof}
\begin{cor}
	The subspace $\mf m_{>-a}$ is isotropic and $\mf n_{>-a}$ is coisotropic in $\mf g^{>-a}$ with respect to (\ref{Def:BFong}). Moreover, $\mf n_{>-a}$ is exactly the orthogonal complement of $\mf m_{>-a}$, i.e., $\mf n_{>-a}=\mf m_{>-a}^\perp:=\{x\in \mf g^{>-a}~|~ (e~|~[x, \mf m_{>-a}])=0\}$ .
\end{cor}
\begin{proof}
Condition (i) in \Cref{Def:AdPare} implies that $(e~|~[\mf m, \mf n])=0$. In particular, $(e~|~[\mf m_{>-a}, \mf n_{>-a}])=0$, hence $\mf m_{>-a}$ is isotropic and $\mf n_{>-a}$ is contained in $\mf m_{>-a}^\perp$. Condition (vi) in \Cref{Def:AdPare} and \Cref{Cor:dimequality} then imply that they must be equal since $\dim\, \mf m_{>-a}+\dim\,\mf m_{>-a}^\perp=\dim\, \mf m_{>-a}+\dim\,\mf n_{>-a}= \dim \, \mf g^{>-a}$. 
\end{proof}

It was proved in \cite{Sadaka} that a  $\Z$-grading $\Gamma: \mf g=\bigoplus_{i\in \Z}\mf g_i$ admits an admissible pair for $e$ if and only if \begin{center} 
	($\ast$): \qquad	$e\in \mf g_a$ for some integer $a>1$ and $\ad\,e$ is injective on $\bigoplus_{i\leq -a}\mf g_i$. \end{center}
It is now very easy to prove this statement.  Namely, let $\mf g^{>-a}$ be a graded complement of $\mf g_{<0}^e$ in $\mf g_{(-a, 0)}$. If $\mf g^{>-a}_{-a/2}\neq 0$, then (\ref{Def:BFong}) restricts to a non-degenerate bilinear form on $\mf g^{>-a}_{-a/2}$. Choose an isotropic subspace $\mf l$ of $\mf g^{>-a}_{-a/2}$ and let $\mf l^\perp$ be its orthogonal complement. Set $$\mf m=\mf g_{\leq -a}\oplus (\mf g^{>-a}\cap \bigoplus_{-a<i<-a/2}\mf g_i)\oplus \mf l, \quad \mf n=\mf g_{\leq -a}\oplus (\mf g^{>-a}\cap \bigoplus_{-a<i<-a/2}\mf g_i)\oplus \mf l^\perp.$$  
Then since $[\mf n, \mf n]\subseteq \mf g_{\leq -a}$, Condition (v) in \Cref{Def:AdPare} is automatically satisfied. All the other conditions can also be easily proved.  \medskip 

Now we can give a characterization of an admissible pair. Let $\Gamma$ be a $\Z$-grading satisfying $(\ast)$ as above. Then the following proposition gives an equivalent definition of an admissible pair with respect to $\Gamma$ and a nilpotent element $e\in \mf g_a$. 

\begin{prop}\label{prop:charofadmpair}
	An admissible pair $(\mf m, \mf n)$ with respect to $e$ is a choice of a graded complement of $\mf g^e_{<0}$ in $\mf g_{(-a, 0)}$, say $\mf g^{>-a}$, and a choice of an isotropic subspace $\mf l$ of $\mf g^{>-a}$,  with $\mf m:=\mf g_{\leq -a}\oplus \mf l, \, \mf n:=\mf g_{\leq -a}\oplus \mf l^\perp$, where $\mf l^\perp$ is the orthogonal complement of $\mf l$ with respect to (\ref{Def:BFong}), such that $\mf n$ is a subalgebra of $\mf g$ and $\mf m$ is an ideal of $\mf n$. 
\end{prop}
\begin{remark}
	Compared to a good $\Z$-grading, it is no longer trivial to show that $\mf n$ is a subalgebra of $\mf g$ and $\mf m$ is an ideal of $\mf n$ for an arbitrary choice of $\mf l$. 
\end{remark}

\subsection{Affine W-algebras associated to admissible pairs} 
Now we replace Condition (iii) in \Cref{Def:AdPare}  by $[\mf n, \mf n]\subseteq \mf m$, i.e., $\mf m$ is an ideal of $\mf n$ containing the derived algebra of $\mf n$. This gives special admissible pairs. Note that all the examples of admissible pairs given in \cite{Sadaka} in fact satisfy this stronger condition, though we expect that this might not be true in general. Under the characterization of \Cref{prop:charofadmpair}, this stronger condition simply means that $[\mf l^\perp, \mf l^\perp]\subseteq \mf m$.

This stronger condition is exactly the condition that we have assumed in the construction of adjusted semi-infinite cohomology when we pass to the affinization. Affine W-algebras associated to admissible pairs can then be defined similarly as in the previous section.  More precisely, let $\tilde{\mf g}$ be the Kac-Moody affinization of $\mf g$, and $V_k(\mf g)$ the vacuum representation. Assume that $(\mf m, \mf n)$ is an admissible pair satisfying $[\mf n, \mf n]\subseteq \mf m$. Consider the nilpotent subalgebra $\hat{\mf n}=\mf n\otimes \C[t, t^{-1}]$ of $\tilde{\mf g}$ and the one-form $\beta_e\in \hat{\mf n}^*$ defined by $\beta_e(u\otimes t^n):=\delta_{n, -1}(e~|~u)$ for $u\in \mf n$. Define $\rho^{\beta_e}$ and $\gamma^{\beta_e}$ as in the previous section. Then $\ker \gamma^{\beta_e}=\hat{\mf m}=\mf m\otimes \C[t, t^{-1}]$. Since $[\hat{\mf n}, \hat{\mf n}]\subseteq \hat{\mf m}$, the adjusted semi-infinite cohomology of $\hat{\mf n}$, with coefficients in $V_k(\mf g)$, with respect to $\beta_e$ can be defined. 
\begin{definition}
	Let  $(\mf m, \mf n)$ be an admissible pair with respect to $e$ satisfying $[\mf n, \mf n]\subseteq \mf m$. The affine W-algebra $ W^k(\mf m, \mf n,  e)$ associated to $(\mf m, \mf n)$ is defined as the adjusted semi-infinite cohomology $H_a^{\infty/2+\bullet}(\hat{\mf n}, \beta_e, V_k(\mf g))$. 
\end{definition}

\begin{remark}
	The subalgebra $\hat{\mf m}$ does not appear in the cohomology.  It  plays the role of the kernel of the 2-cocyle $\gamma^{\beta_e}$ in the construction of adjusted semi-infinite cohomology.
\end{remark}
The vertex algebra structure on the adjusted semi-infinite cohomology comes from the fact that the adjusted Feigin complex is the tensor product of three vertex (super)algebras, and the differential  is the zero mode of some odd element. One can use the same method as in the appendix of \cite{Kac&Alberto2006} to show that finite W-algebras associated to admissible pairs defined in \cite{Sadaka} are Zhu algebras of affine W-algebras associated to the same admissible pairs.

\bibliographystyle{alpha}

\end{document}